\documentclass[journal,10pt]{IEEEtran}

\usepackage{xcolor}
\usepackage{relsize}

\usepackage{ifpdf}
\usepackage{graphicx}
\usepackage{cite}
\usepackage{nomencl}
\usepackage{enumitem}
\usepackage{slashbox}
\usepackage[normalem]{ulem}

\ifCLASSINFOpdf
\else
\fi
\usepackage{amsmath,amssymb,bm}
\usepackage{algorithmic}
\usepackage{adjustbox}
\usepackage{array}
\usepackage[tablename=Table]{caption}
\usepackage{stfloats}
\usepackage{url}
\usepackage{amsthm}
\usepackage{multirow}

\newtheorem{corollary}{Corollary}

\newcommand{\RomanNumeralCaps}[1]{\MakeUppercase{\romannumeral #1}}
\usepackage{stackengine,scalerel}

\newcommand\overstar[1]{\ThisStyle{\ensurestackMath{%
			\setbox0=\hbox{$\SavedStyle#1$}%
			\stackengine{0pt}{\copy0}{\kern.2\ht0\smash{\SavedStyle*}}{O}{c}{F}{T}{S}}}}

\newcommand{\quotes}[1]{``#1''}

\begin{document}
 
\title{Spatio-Temporal Probabilistic Voltage Sensitivity Analysis - A Novel Framework for Hosting Capacity Analysis}

\author {\IEEEauthorblockN{Sai Munikoti,~\textit{Student Member, IEEE}, Mohammad Abujubbeh,~\textit{Student Member, IEEE}, \\ Kumarsinh Jhala,~\textit{Member, IEEE}, Balasubramaniam Natarajan,~\textit{Senior Member, IEEE}}}


\maketitle

\begin{abstract}
Smart grids are envisioned to accommodate high penetration of distributed photovoltaic (PV) generation, which may cause adverse grid impacts in terms of voltage violations. Therefore, PV Hosting capacity (HC) is being used as a planning tool to determine the maximum PV installation capacity that causes the first voltage violation and above which would require infrastructure upgrades. Traditional methods of HC analysis are computationally complex as they are based on iterative load flow algorithms that require investigation of a large number of scenarios for accurate assessment of PV impacts. This paper first presents a computationally efficient analytical approach to compute the probability distribution of voltage change at a particular node due to random behavior of randomly located multiple distributed PVs. Next, the derived distribution is used to identify voltage violations for various PV penetration levels and subsequently determine the HC of the system without the need to examine multiple scenarios. Results from the proposed spatio-temporal probabilistic voltage sensitivity analysis and the HC are validated via conventional load flow based simulation approach on the IEEE 37 and IEEE 123 node test systems.\footnote{This work has been submitted to the IEEE for possible publication. Copyright may be transferred without notice, after which this version may no longer be accessible.}
 
\end{abstract}


\IEEEpeerreviewmaketitle
\vspace{-0.5cm}
\section{Introduction}

\IEEEPARstart{P}{ower} grid is undergoing significant changes to meet modern-day energy demand in a more efficient manner. Integration of renewable energy sources, especially rooftop Photovoltaics (PVs), is a potential solution to (1) lower the overall carbon footprint; (2) lower operational cost; (3) provide various ancillary services such as peak load shaving; and (4) restore voltage services to critical loads during contingencies. Therefore, many countries are aiming to meet a major portion of energy demand through renewable energy sources. For example, by 2050, USA, China, EU and India are projected to have $63$\%, $67$\%, $70$\% and $73$\% of their total energy use met through renewables, respectively \cite{irena2018global}. Despite the aforementioned benefits, high penetration levels of PVs may impact the grid negatively in terms of voltage fluctuations and stability. This necessitates the need for thoroughly analyzing the grid in the presence of PVs to maximize their integration benefits. In this regard, PV hosting capacity (HC) as a planning tool has attracted the attention of many researchers and practitioners. HC refers to the maximum amount of PV generation that can be accommodated in the distribution system, while keeping system operational constraints within their safe limits, without the need for infrastructure upgrades.

A comprehensive HC analysis monitors power quality, power loss, thermal overload, protection devices, and voltage deviation for different PV penetration levels. With increasing PV penetration levels in distribution systems, many operational issues have emerged, including voltage violations \cite{alyami2014adaptive,olivier2015active} which directly impact the HC. Therefore, the development of an accurate, yet computationally efficient HC analysis framework is essential to ensure efficient, economic, and reliable operation of the distribution system. Most of the existing methods for evaluating HC are simulation-based and require the execution of multiple load-flow runs for various PV allocation scenarios. The drawbacks of simulation-based HC studies are (1) high computational complexity, which increases with the size of the network; (2) scenario-based and scenario-specific results, which do not provide any performance guarantee for a more general case; and 3) very conservative results, typically based on worse case scenarios. As voltage is the primary concern for many utilities while determining the HC \cite{tonkoski2012impact,ding2016distributed}, voltage sensitivity analysis (VSA) can help identify voltage violations, which in turn can be used to compute the HC of the system. Traditional methods of VSA such as load flow based numerical approaches have been used by various researchers \cite{smith2012stochastic,dubey2016estimation, rylander2015streamlined} to compute the HC. As numerical VSA methods are simulation-based, they possess the same demerits as monte-carlo based power flow methods. Additionally, at the planning stage, there could be scenarios where the grid operator may not be aware of the actual PV locations in the network. Under these scenarios of random power change at unknown locations of the grid, numerical VSA approaches would involve simulation of a large number of scenarios to account for temporal and spatial uncertainties associated with power change. This would result in high computational complexity. To overcome the drawbacks of numerical methods, there are some limited analytical approaches for VSA that have been proposed in \cite{8017483, 8973956, Munikoti2020probabilistic}. The analytical VSA paradigm enables quick, yet accurate, estimation of node vulnerability to voltage violations. The powerfulness of this paradigm lies in its ability to account for spatio-temporal uncertainties associated with power change at system nodes in a computationally efficient manner.
This paper proposes a novel analytical spatio-temporal probabilistic voltage sensitivity analysis (ST-PVSA) framework to derive the probability distribution of voltage change for an unbalanced distribution system that incorporates uncertainties associated with both power change (intermittent PV generation) and the locations of PV units in the system. This fundamental theoretical result is then used to simplify HC computation. 

\subsection{Literature review}

Zain \textit{et. al.} in \cite{zain2020review} have reviewed the literature related to HC and classified the efforts into four major categories (deterministic, stochastic, optimization-based, and streamlined) based on the available data and the type of study to be performed. Most of the existing approaches for HC depend on numerical load flow-based methods \cite{smith2012stochastic,ding2016distributed, dubey2016estimation, rylander2015streamlined}, and involve the analysis of multiple PV deployment scenarios. For instance, in \cite{ding2016distributed}, a scenario is generated by randomly allocating PVs in the network, and then load flow is executed for each penetration level until a voltage violation is encountered. To cover all possible locations, the complete process is repeated for multiple scenarios, thereby presenting a huge computational burden. \cite{dubey2016estimation} assigns each feeder a minimum and maximum HC, corresponding to the most conservative and most optimistic HC value. However, the overvoltage risk within the range of two HC endpoint values is not quantified. Similarly, authors in \cite{jain2019quasi}, propose a quasi-static-time-series (QSTS) based dynamic PV HC methodology. Here, power flow analysis is conducted on the load and PV data over one year, where the time duration of violation is also monitored along with the total violations count. For a real distribution model with thousands of nodes and one-second resolution data, an annual simulation could take a few days \cite{jain2019quasi}.
Furthermore, the PV and load uncertainties have significant influences upon HC values. As a result, probabilistic HC methods have gained attention \cite{xu2014probabilistic, valverde2018estimation}. Though the probabilistic approaches can effectively describe the uncertainty in fluctuations of PVs and loads, most of these approaches are simulation-based and thus are computationally inefficient. More importantly, the performance of these approaches relies heavily on the availability of data. 

In addition to HC, numerical VSA methods have also been used to guide various grid applications such as voltage regulation, DER allocations, etc. \cite{yan2012voltage, newaz2019coordinated, kang2019reactive}. For instance, authors in \cite{yan2012voltage} propose a method for analyzing voltage variations due to PV generation fluctuations, considering a variety of factors. However, its dependency on the inefficient simulation method limits its applications to large scale distribution networks. Similarly, authors in \cite{kang2019reactive} develop an optimization model for the electric vehicle charging schedule based on VSA approaches. Still, the requirements of iterative executions of power flow calculations and optimization models hinder its application in real-world scenarios. Thus, traditional methods suffer from high computational complexity and do not provide analytical insight into the underlying physics of the system. 
Therefore, to overcome the drawbacks of numerical methods, there are some limited analytical approaches for VSA that have been proposed. Authors in \cite{brenna2010voltage,zad2016centralized}, develop an algorithm based on VSA which optimally manages active and reactive powers of DGs to keep the system voltages inside the limits. Here, instead of repeating load flow calculations to solve the optimization problem, a sensitivity matrix is used to conduct load flow computation in a non-iterative manner, reducing the computational burden significantly. However, the algorithms proposed are not properly validated with standard test systems. Authors in \cite{klonari2016application}, have taken a probabilistic approach where smart meter data is used along with sensitivity analysis to define boundary values of various operation indices. This approach does not account for unbalanced load conditions. In \cite{8973956, Munikoti2020probabilistic}, authors have developed a new probabilistic voltage sensitivity approach (PVSA) to quantify voltage change in a computationally efficient way and accounts for the temporal uncertainties associated with random power change at fixed locations of the grid. However, the assumption of a fixed location is not generic enough to account for unknown locations of PV installations. Our prior work \cite{jhala2018probabilistic} on PVSA provides a quick and efficient tool for estimating the distribution of voltage change in a balanced distribution system due to spatial randomness in PV installations.

To summarize, probabilistic VSA represents a viable, low complexity, systematic approach to HC computation. Therefore, in this work, a generic computationally efficient analytical framework for probabilistic voltage sensitivity is developed for an unbalanced distribution system that systematically accounts for the spatio-temporal uncertainties associated with PV generation. This spatio-temporal probabilistic voltage sensitivity analysis (ST-PVSA) framework is then employed to determine the HC of an unbalanced power distribution system. The proposed formulation helps grid operators prepare for the future modernized grid by providing new insights on the impact of high PV penetration on grid operations.

\subsection{Contributions}
This work proposes a novel stochastic method of grid voltage sensitivity analysis which is used to calculate the likelihood of node voltage exceeding operational bounds. This framework is then employed to determine the PV HC of an unbalanced distribution system. The key contributions of this paper are listed below:
\begin{itemize}
    \item The probability distribution of voltage change due to random change in complex power across random locations of the distribution grid is derived analytically.
\item Unlike \cite{jhala2018probabilistic}, which is only valid for balanced networks, the proposed framework is generic enough to work for both balanced and unbalanced distribution systems.

\item The proposed ST-PVSA method is used to (1) analyze the aggregate effect of spatially random distribution of PVs on the feeder voltage, and (2) determine the probability of node voltage exceeding allowable limits. Analytical results are validated using simulation results on the IEEE 37-node test system.

\item The proposed ST-PVSA is employed to determine PV HC in significantly less time with accurate results compared to existing load flow-based approaches.
\end{itemize}
The rest of the paper is organized as follows: Section \RomanNumeralCaps{2} discusses the typical simulation-based method for HC. Then, the probability distribution of voltage change with spatio-temporal uncertainties is derived in Section \RomanNumeralCaps{3}, followed by its validation in Section \RomanNumeralCaps{4}. In Section \RomanNumeralCaps{5}, the derived distribution is used to determine the HC and validated with a load flow-based approach, and finally, conclusions are provided in Section \RomanNumeralCaps{6}.
\vspace{-0.4cm}
\section{HC with simulation-based approach}
\begin{figure}[t]
\centering
	\includegraphics[width = 8.5cm, height=6.0cm]{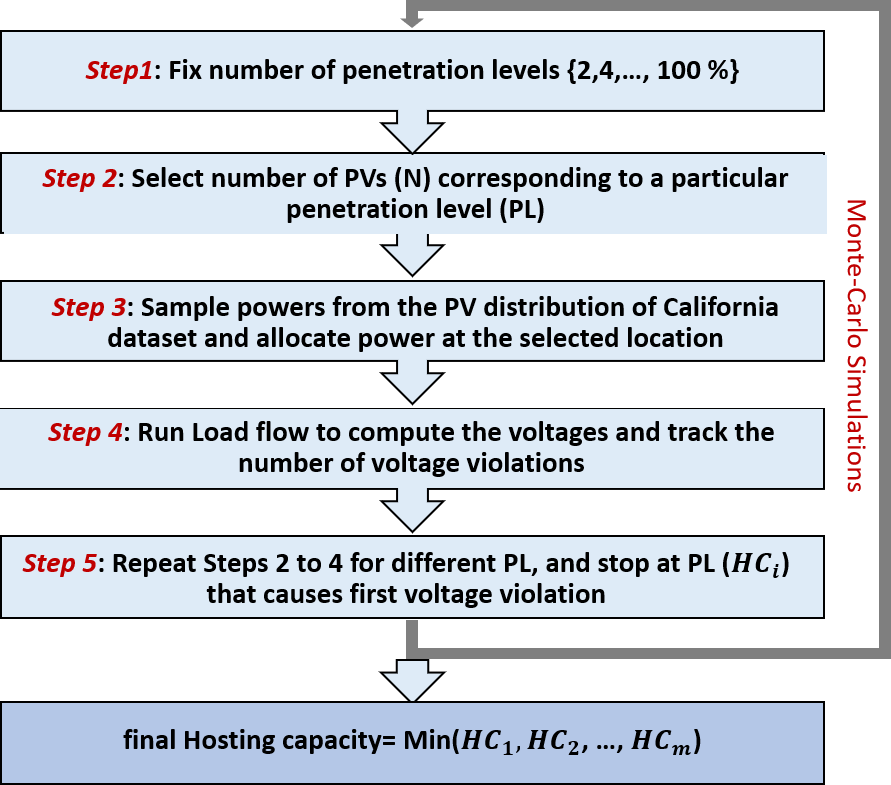}
	\caption{Flowchart of Load flow based HC method }
	\label{fig:3}
\end{figure}

This section describes a typical load flow based approach of determining the HC of the system. Here, the net power injection is increased in steps by allocating power to PVs located at random locations of the network. Then, load flow is executed for each penetration levels to track the number of node voltage violations throughout the network. This process is repeated for multiple penetration levels until the number of violations exceeds the threshold. The corresponding power (penetration level) is the HC for a particular PV deployment scenario. Thereafter, the complete process is repeated multiple times to cover all possible spatial distribution of the PV installations and the minimum capacity across all such scenarios is the final HC of the network. The scenario based analysis presents a huge computational burden due to the requirement of multiple load flow runs. Fig. \ref{fig:3} depicts the flow chart of the existing load flow based approach of computing HC \cite{ding2016distributed}.

Alternatively, this paper attempts to develop a probabilistic VSA approach that determines the HC in a computationally efficient manner. As mentioned in \cite{smith2012stochastic}, a comprehensive analysis of PV distribution needs to monitor voltage, protection, power quality and control limits. However, voltage is the primary concern for many utilities \cite{tonkoski2012impact,ding2016distributed}. So, similar to \cite{ding2016distributed, jain2019quasi}, this paper only considers voltage limits to determine the PV HC. The first step towards the probabilistic VSA approach for HC is to derive an analytical expression of voltage sensitivity due to random power change at random locations in the network, as presented in the next section.

\section{ST-PVSA for random distribution of PVs}

This section details the steps involved in the derivation of the probability distribution of voltage change at network nodes due to random power changes at random locations of the network. Throughout this paper, observation nodes are referred to those nodes where voltage change is observed and actor nodes are those where power changes. In our fundamental work on analytical VSA \cite{Munikoti2020probabilistic}, we derive an analytical approximation of voltage change at any observation node due to power change at multiple actor nodes for an unbalanced distribution system. This result is presented in Corollary 1 for completeness.
\begin{corollary}
	For an unbalanced power distribution system, change in complex voltage $\Delta V_{O}$ at an observation node ($O$) due to change in complex power at multiple actor nodes can be approximated by  
	\begin{equation}	
	\begin{bmatrix}
	\Delta V_{O}^{a} \\[15pt]
	\Delta V_{O}^{b} \\[15pt]
	\Delta V_{O}^{c} 
	\end{bmatrix} \approx- \sum_{A \epsilon \tilde{A}}  \left(	
	\begin{bmatrix}
\frac{\Delta S_{A}^{a\star}Z_{OA}^{aa}}{ V_{A}^{a\star}} + \frac{\Delta S_{A}^{ b\star}Z_{OA}^{ab}}{ V_{A}^{b\star}}+ \frac{\Delta S_{A}^{ c\star}Z_{OA}^{ac}}{ V_{A}^{c\star}} \\[8pt]
\frac{\Delta S_{A}^{ a\star}Z_{OA}^{ba}}{ V_{A}^{a\star}} + \frac{\Delta S_{A}^{ b\star}Z_{OA}^{bb}}{ V_{A}^{b\star}}+ \frac{\Delta S_{A }^{c\star}Z_{OA}^{bc}}{ V_{A}^{c\star}}  \\[8pt]
\frac{\Delta S_{A}^{a\star}Z_{OA}^{ca}}{ V_{A}^{a\star}} + \frac{\Delta S_{A}^{b\star}Z_{OA}^{cb}}{ V_{A}^{b\star}}+ \frac{\Delta S_{A}^{c\star}Z_{OA}^{cc}}{ V_{A}^{c\star}}
\end{bmatrix}\right)
\label{eq:2} 
\end{equation}
\end{corollary}
\begin{proof}
See \cite{Munikoti2020probabilistic}.
\end{proof}
In (\ref{eq:2}), superscript $a,b$ and $c$ represent the three phases; this notation is used throughout the paper. $V_{A}^{a\star}$ and $\Delta S_{A}^{a}$ represent complex conjugate of voltage at phase $a$ and complex power change at actor node $A$, respectively; $Z_{OA}$ denotes the impedance matrix including self and mutual line impedance of the shared path between observation node and actor node from the source node. 
The set $\tilde{A}$ contains all actor nodes.
Building on this fundamental result, this work develops ST-PVSA to derive the probability distribution of voltage change at any given node due to random spatial distribution of PV installations.

The change in complex voltage at any phase (say phase $a$) of observation node $O$ due to change in complex power at any phase of a single actor node $A$ is given by,
\begin{equation}
\Delta V_{OA}^{a} = \frac{\Delta S_{A}^{a\star}Z_{OA}^{aa}}{ V_{A}^{a\star}} + \frac{\Delta S_{A}^{ b\star}Z_{OA}^{ab}}{ V_{A}^{b\star}}+ \frac{\Delta S_{A}^{ c\star}Z_{OA}^{ac}}{ V_{A}^{c\star}}.
\label{eq:3a}
\end{equation}
The complex voltage change can be decomposed into real and imaginary parts as,
\begin{equation}
   \Delta V_{OA}^{a}=\Delta V_{OA}^{a,r} + j\Delta V_{OA}^{a,i}. 
   \label{eq:3}
\end{equation}
For simplicity, the voltage change expression throughout the derivation is shown for a single phase (phase $a$). However, similar form and approach is applicable to other phases as well. On expanding power change ($\Delta S_{A}^{a\star}=\Delta P-j\Delta Q$) and impedance ($Z_{OA}=R+jX$) components in (\ref{eq:3a}), the real ($\Delta V_{OA}^{a,r}$) and imaginary parts ($\Delta V_{OA}^{a,i}$) of voltage change at phase $a$ of the observation node $O$ can be written as,

\begin{equation*}
\begin{split}
\Delta V_{OA}^{a,r} = \sum_{h,u}^{}\frac{-1}{|V_{A}^{h}|}[\Delta P_{A}^{h}(R_{OA}^{u}cos(\omega_{A}) - X_{OA}^{u}sin (\omega_{A})) \\ 
+ \Delta Q_{A}^{h}(R_{OA}^{u}sin(\omega_{A}) + X_{OA}^{u} cos(\omega_{A})) ] \\
\Delta V_{OA}^{a,i} = \sum_{h,u}^{}\frac{-1}{|V_{A}^{h}|}[
\Delta P_{A}^{h}(R_{OA}^{u}sin(\omega_{A}) + X_{OA}^{u} cos(\omega_{A})) \\
+\Delta Q_{A}^{h}(X_{OA}^{u}sin (\omega_{A})-R_{OA}^{u}cos(\omega_{A}))]
\end{split}
\label{eq:4}
\end{equation*} 

where $h$ $\epsilon$ $\tilde{H}$ and $u$ $\epsilon$ $\tilde{U}$. The set $\tilde{H}=\{a,b,c\}$ denotes different phases and the set $\tilde{U}=\{ aa,ab,ac\}$ represents phase sequence for the corresponding phase. $\Delta P_{A}^{h}$ and $ \Delta Q_{A}^{h}$ are the active and reactive power changes, respectively. $R_{OA}^{h}$ and $ X_{OA}^{h}$ are the resistance and reactance of shared path between the observation node $O$ and actor node $A $ from the source node. $V_{A}^{h} $ denotes the complex rated voltage of actor node $A$. 
The magnitude and angle of voltage at a particular phase, say phase $a$, of node $A$ are given by $|V_{A}^{a}|$ and $\theta_{A}^{a}$, respectively with reference to the slack bus. Line voltage of the network is always kept within permissible limits, and thus it is reasonable to assume the phase difference of $120^{\circ}$ between the voltage angles of different phases with the same angle for all the node voltages of each phase. Based on this assumption, $\Delta V_{OA}^{a,r}$, $\Delta V_{OA}^{a,i}$ can be rewritten as:
\begin{equation}
\footnotesize
\begin{split}
\Delta V_{OA}^{a,r} =
\frac{-\Delta P_{A}^{a}R_{OA}^{aa}}{|V_{A}^{a}|} + 
\frac{\Delta P_{A}^{b}}{|V_{A}^{b}|}
\left(\frac{R_{OA}^{ab}}{2}-
\frac{\sqrt{3}X_{OA}^{ab}}{2} \right) + \\ 
\frac{\Delta P_{A}^{c}}{|V_{A}^{c}|}
\left(\frac{R_{OA}^{ac}}{2} + 
\frac{\sqrt{3}X_{OA}^{ac}}{2} \right) -
\frac{\Delta Q_{A}^{a}X_{OA}^{aa}}{|V_{A}^{a}|} + \\ 
\frac{\Delta Q_{A}^{b}}{|V_{A}^{b}|}\left(\frac{\sqrt{3}R_{OA}^{ab}}{2}+\frac{X_{OA}^{ab}}{2}\right) + 
\frac{\Delta Q_{A}^{c}}{|V_{A}^{c}|}\left(\frac{-\sqrt{3}R_{OA}^{ac}}{2}+\frac{X_{OA}^{ac}}{2}\right)
\label{eq:5}
\end{split}
\end{equation}

\begin{equation}
\footnotesize
\begin{split}
\Delta V_{OA}^{a,i} =\frac{-\Delta P_{A}^{a}X_{OA}^{aa}}{|V_{A}^{a}|} + 
\frac{\Delta P_{A}^{b}}{|V_{A}^{b}|}\left(\frac{\sqrt{3}R_{OA}^{ab}}{2}+\frac{X_{OA}^{ab}}{2}\right) +  \\ 
\frac{\Delta P_{A}^{c}}{|V_{A}^{c}|}\left(\frac{-\sqrt{3}R_{OA}^{ac}}{2}   +\frac{X_{OA}^{ac}}{2} \right) +
\frac{\Delta Q_{A}^{a}R_{OA}^{aa}}{|V_{A}^{a}|} + \\ 
\frac{\Delta Q_{A}^{b}}{|V_{A}^{b}|}\left(-\frac{R_{OA}^{ab}}{2}+\frac{\sqrt{3}X_{OA}^{ab}}{2}\right) + 
\frac{\Delta Q_{A}^{c}}{|V_{A}^{c}|}\left(-\frac{R_{OA}^{ac}}{2}-\frac{\sqrt{3}X_{OA}^{ac}}{2}\right)
\label{eq:6}
\end{split}
\end{equation}
The real (\ref{eq:5}) and imaginary (\ref{eq:6}) parts of the voltage change can further be represented in a simplified form as,
\begin{equation}
\Delta V_{OA}^{a,r}= \mathbf{(Z^{a,r})^T\Delta S}, \hspace{1cm} \Delta V_{OA}^{a,i}= \mathbf{(Z^{a,i})^T\Delta S} 
\label{eq:6b}
\end{equation}

\begin{minipage}{.1\textwidth}
\begin{align*}
\begin{split}
\footnotesize
\mathbf{Z^{a,r}}= \begin{bmatrix}
    -R_{OA}^{aa}\\[1pt]
	\frac{R_{OA}^{ab}}{2}-\frac{\sqrt{3}X_{OA}^{ab}}{2} \\[1pt]
    \frac{R_{OA}^{ac}}{2}+\frac{\sqrt{3}X_{OA}^{ac}}{2} \\[1pt]
    -X_{OA}^{aa} \\[1pt]
    \frac{\sqrt{3}R_{OA}^{ab}}{2}+\frac{X_{OA}^{ab}}{2} \\[1pt]
    \frac{-\sqrt{3}R_{OA}^{ac}}{2}+\frac{X_{OA}^{ac}}{2} \\[1pt]
	\end{bmatrix}
\end{split}
\end{align*}
\end{minipage}
\begin{minipage}{.1\textwidth}
\begin{align*}
\footnotesize
\mathbf{Z^{a,i}} = \begin{bmatrix}
    -X_{OA}^{aa}\\[1pt]
	\frac{\sqrt{3}R_{OA}^{ab}}{2}+\frac{X_{OA}^{ab}}{2} \\[1pt]
   \frac{-\sqrt{3}R_{OA}^{ac}}{2}+\frac{X_{OA}^{ac}}{2} \\[1pt]
    R_{OA}^{aa} \\[1pt]
    \frac{-R_{OA}^{ab}}{2}+\frac{\sqrt{3}X_{OA}^{ab}}{2} \\[1pt]
   \frac{-R_{OA}^{ac}}{2}-\frac{\sqrt{3}X_{OA}^{ac}}{2} \\[1pt]
	\end{bmatrix}
\end{align*}
\end{minipage}\\
\vspace{0.15cm}
\hspace{0.8cm} $\mathbf{\Delta S}= \left[\frac{ \Delta P_{A}^{a}}{|V_{A}^{a}|}
        \hspace{0.1cm}\frac{ \Delta P_{A}^{b}}{|V_{A}^{b}|}
       \hspace{0.1cm} \frac{ \Delta P_{A}^{c}}{|V_{A}^{c}|}
       \hspace{0.1cm} \frac{ \Delta Q_{A}^{a}}{|V_{A}^{a}|}
        \hspace{0.1cm}\frac{ \Delta Q_{A}^{b}}{|V_{A}^{b}|}
        \hspace{0.1cm}\frac{ \Delta Q_{A}^{c}}{|V_{A}^{c}|}\right]$ \\
where $\mathbf{Z^{a,r}}$ and $\mathbf{Z^{a,i}}$ are the vectors incorporating shared path impedance terms corresponding to real and imaginary parts of voltage change, respectively. To represent the random variation of PV generation, the real and reactive power change is modeled as a random variable. Consistent with the prior efforts in modeling PV generation as a time series with a trend component  and  Gaussian noise \cite{hassanzadeh2010practical, vasilj2015pv, jhala2019data}, the power variation is assumed to be Gaussian. It is important to note that the framework is quite general to account for any arbitrary random variable with finite mean and variance. Therefore, the vector $\mathbf{\Delta S}$, which incorporate the terms corresponding to the ratio of power change and constant base voltages, can be expressed as Gaussian random vector $ \mathbf{\Delta S} {\sim} \mathcal{N} (\mu_{\Delta S}, \textstyle\sum_{\Delta S})$ with $\bm{\mu_{\Delta S}}$ being mean vector, and covariance matrix $\bm{\textstyle\sum_{\Delta S}}$ as,
\vspace{-0.15cm}
\begin{equation*}
\footnotesize
\begin{bmatrix}
 \sigma^{2}_{\frac{P^{a}}{|V_{A}^{a}|}} \hspace{0.2cm} cov( \Delta P_{A}^{a}/|V_{A}^{a}|, \Delta P_{A}^{b}/|V_{A}^{b}|) \hdots cov( \Delta P_{A}^{a}/|V_{A}^{a}|, \Delta Q_{A}^{c}/| V_{A}^{c}|) \\
 \vdots \\ 
 \hspace{-0.05cm} cov( \Delta P_{A}^{a}|/V_{A}^{a}|, 
 \Delta Q_{A}^{c} / |V_{A}^{c}|) \hspace{2cm} \hdots \hspace{1.5cm} \sigma^{2}_{\frac{Q^{c}}{|V_{A}^{c}|}} 
 \end{bmatrix}
 \label{eq:8}
\end{equation*}
Here, the diagonal and off-diagonal elements indicate variance and covariance among the terms that are ratio of power changes and base voltages across different phases of actor nodes, respectively. The impedance of the shared line between a given observation node ($O$) and a random actor node can be modeled as a correlated random variable. The mean, variance and covariance of resistance $R_{OA}$ and reactance $X_{OA}$ corresponding to a given observation node $O$ can be estimated based on actual line impedance data. In addition, let $\bm{\mu_{Z^{r}}}$ and $\bm{\mu_{Z^{i}}}$ represent the mean of real ($\bm{Z^{a,r}}$) and imaginary ($\bm{Z^{a,i}}$) parts of impedance vector, respectively. The average is taken over all the nodes of the network with respect to the observation node. Similarly, $\bm{\textstyle\sum_{Z^{r}}}$ and $\bm{\textstyle\sum_{Z^{i}}}$ denote the covariance matrices of $\bm{Z^{a,r}}$ and $\bm{Z^{a,i}}$, respectively. The correlation coefficient between the shared path impedances for various actor nodes is computed based on network parameters. Particularly, the objective of this work is to derive the probability distribution of the magnitude of voltage change at an observation node due to random power variation of PVs located at random nodes, which will further be used to estimate the system HC. The probability distribution of real $\Delta V_{OA}^{a,r}$ and imaginary components $\Delta V_{OA}^{a,i}$ of the voltage change due to random spatial distribution of multiple PV units can be derived using the following steps:
\vspace{0.2cm} \\
\textit{\textbf{Step 1:} Compute mean and variance of $\Delta V_{OA}^{a,r}$ and $\Delta V_{OA}^{a,i}$ due to a single actor node}:\\
Using (\ref{eq:6b}), the mean of the voltage change can be expressed as the expectation of product of two terms, i.e., the shared path impedance vector ($\bm{Z^{a,r}}$ for real and $\bm{Z^{a,i}}$ for imaginary part) and power change vector $\bm{\Delta S}$. As the terms in the product are mutually independent, the expectation of their product can be applied to individual terms separately yielding the mean of real ($\mu_{r}$) and imaginary ($\mu_{i}$) parts as,
\begin{equation}
\begin{split}
    \mu_{r} = E[Z^{(a,r)^{T}}\Delta S]= E[Z^{(a,r)^{T}}]E[\Delta S]=  \mathlarger{\bm{\mu_{Z_{o}^{a,r}}}}
    \mathlarger{\bm{\mu_{\Delta S}}} \\
    \mu_{i} = E[Z^{(a,i)^{T}}\Delta S]= E[Z^{(a,i)^{T}}]E[\Delta S]= \mathlarger{\bm{\mu_{Z_{o}^{a,i}}}} \mathlarger{ \bm{\mu_{\Delta S}}}
    \label{eq:9}
\end{split}    
\end{equation}
Furthermore, the variance of real and imaginary parts
of the voltage change can be computed as shown below,
\begin{equation}
\begin{split}
    \text{Var}(\Delta V_{OA}^{a,r})= E[(Z^{(a,r)^{T}}\Delta S)^2]- E[(Z^{(a,r)^{T}}\Delta S)]^2 \\
    \text{Var}(\Delta V_{OA}^{a,i})= E[(Z^{(a,i)^{T}}\Delta S)^2]- E[(Z^{(a,i)^{T}}\Delta S)]^2.    
    \label{eq:11}
    \end{split}
\end{equation}
Since $Z_{r}^T$ and $\Delta S$ are independent, the expectation of their product can be written in terms of product of their individual expectation as, 
\begin{equation}
   E[Z^{(a,r)^T} \Delta S \Delta S^T Z^{a,r}] - (E[Z^{(a,r)^T}]E[\Delta S])^2. \label{eq:12}
\end{equation}
For simplicity, the equation for variance is shown for the real part of voltage change and a similar form exists for imaginary part. Now, using the properties of matrix trace, the variance of the real part can be rewritten as,
\begin{equation}
\begin{split}
    E[Tr(Z_{r}Z_{r}^T\Delta S \Delta S^T)] - (\mu_{Z^r}\mu_{\Delta S})^2 \\
    =Tr(E[Z_{r}Z_{r}^T]E[\Delta S \Delta S^T]) - (\mu_{Z^r}\mu_{\Delta S})^2 \\
=Tr[(\mu_{Z^r} \mu_{Z^r}^T+ \textstyle\sum_{Z^{r}})(\mu_{\Delta S} \mu_{\Delta S}^T+ \textstyle\sum_{\Delta S})] - (\mu_{Z^r}\mu_{\Delta S})^2 \\
    =Tr(\mu_{Z^r} \mu_{Z^r}^T \mu_{\Delta S} \mu_{\Delta S}^T ) +Tr(\mu_{Z^r} \mu_{Z^r}^T \textstyle\sum_{\Delta S}) + \\
    Tr(\textstyle\sum_{Z^{r}} \mu_{\Delta S} \mu_{\Delta S}^T) + Tr(\textstyle\sum_{Z^{r}} \textstyle\sum_{\Delta S}) - (\mu_{Z^r}^T \mu_{\Delta S})^2 \\
\end{split}
\label{eq:13}
\end{equation}
Now, the term $Tr(\mu_{Z^r} \mu_{Z^r}^T \mu_{\Delta S} \mu_{\Delta S}^T)$ is rearranged to $ (\mu_{Z^r}\mu_{\Delta S})^2 $, that cancels the last term of eqn. (\ref{eq:13}). After applying trace operator, the variance of the real part can be expressed as,
\begin{equation}
\bm{\mu_{Z^r}^T \textstyle\sum_{\Delta S} \mu_{Z^r}} + \bm{\mu_{\Delta S}^T \textstyle\sum_{Z^{r}} \mu_{\Delta S}} + Tr\bm{(\textstyle\sum_{Z^{r}} \textstyle\sum_{\Delta S})} 
\label{eq:15}
\end{equation}
Following the same steps from equations (\ref{eq:12}- \ref{eq:15}), the variance of imaginary part of voltage change can be written as,
\begin{equation}
\bm{\mu_{Z^i}^T \textstyle\sum_{\Delta S} \mu_{Z^i}} + \bm{\mu_{\Delta S}^T \textstyle\sum_{Z^{i}} \mu_{\Delta S}} + Tr\bm{(\textstyle\sum_{Z^{i}} \textstyle\sum_{\Delta S})}
\label{eq:16}
\end{equation}\\
\textit{\textbf{Step 2:} Compute covariance between real $\Delta V_{OA}^{a,r}$ and imaginary $\Delta V_{OA}^{a,i}$ parts of voltage change}:\\
The covariance between the real and imaginary parts of voltage change can be expressed as:
\begin{equation*}
\small
\begin{split}
  Cov(\Delta V_{OA}^{a,r}, \Delta V_{OA}^{a,i}) &=  E(\Delta V_{OA}^{a,r} \Delta V_{OA}^{a,i})-E(\Delta V_{OA}^{a,r})E(\Delta V_{OA}^{a,i})\\
  &= E[Z_{A}^{(a,r)^T}\Delta S_{A} Z_{A}^{(a,i)^T}\Delta S_{A}]
\end{split}
\label{eq:17b}
\end{equation*}
$\bm{Z^{(a,r)^T}\Delta S}$ and $\bm{Z^{(a,i)^T}\Delta S}$ are expanded using eqn. (\ref{eq:5}) and (\ref{eq:6}) as, to express covariance as the expectation of following term, 
\begin{equation*}
\tiny
\begin{split}
\Bigg[ \Bigg.
& \Bigg\{ \Bigg.
\frac{\Delta P_{A}^{a}}{|V_{A}^{a}|}(-R_{OA}^{aa}) + \frac{\Delta P_{A}^{b}}{| V_{A}^{b}|}\left(\frac{R_{OA}^{ab}}{2}-\frac{\sqrt{3}X_{OA}^{ab}}{2}\right) +
\frac{\Delta P_{A}^{c}}{| V_{A}^{c}|}\left(\frac{R_{OA}^{ac}}{2}+\frac{\sqrt{3}X_{OA}^{ac}}{2}\right) \\ & 
-\frac{\Delta Q_{A}^{a}}{| V_{A}^{a}|}X_{OA}^{aa} + 
\frac{\Delta Q_{A}^{b}}{| V_{A}^{b}|}\left(\frac{\sqrt{3}R_{OA}^{ab}}{2}+\frac{X_{OA}^{ab}}{2}\right) + 
\frac{\Delta Q_{A}^{c}}{| V_{A}^{c}|}\left(-\frac{\sqrt{3}R_{O1}^{ac}}{2}+\frac{X_{O1}^{ac}}{2}\right) 
\Bigg. \Bigg\} \times \\ 
& \Bigg\{ \Bigg.
\frac{\Delta P_{A}^{a}}{| V_{A}^{a}|}(-X_{OA}^{aa}) + \frac{\Delta P_{A}^{b}}{| V_{A}^{b}|}\left(\frac{\sqrt{3}R_{OA}^{ab}}{2}+\frac{X_{OA}^{ab}}{2}\right) + \frac{\Delta P_{A}^{c}}{| V_{A}^{c}|}\left(\frac{-\sqrt{3}R_{OA}^{ac}}{2}+\frac{X_{OA}^{ac}}{2}\right) \\ & 
+\frac{\Delta Q_{A}^{a}}{| V_{A}^{a}|}R_{OA}^{aa} + 
\frac{\Delta Q_{A}^{b}}{| V_{A}^{b}|}\left(-\frac{R_{OA}^{ab}}{2}+\frac{\sqrt{3}X_{OA}^{ab}}{2}\right) + 
\frac{\Delta Q_{A}^{c}}{| V_{A}^{c}|}\left(-\frac{R_{OA}^{ac}}{2}-\frac{\sqrt{3}X_{OA}^{ac}}{2}\right)
\Bigg. \Bigg\} 
\Bigg. \Bigg]
\end{split}
\label{eq:18a}
\end{equation*}
The terms inside the expectation operator are cross multiplied to yield the final expression as,
\begin{equation}
\footnotesize
\begin{split}
\frac{\rho_{p^{a}}\sigma_{p^{a}}^{2}}{| V_{A}^{a}|^{2}} \mu_{R^{aa}}\mu_{X^{aa}} - \frac{\rho_{q^{a}}}{| V_{A}^{a}|^{2}} \sigma_{q^{a}}^{2}\mu_{R^{aa}}\mu_{X^{aa}}+\\ 
\frac{\rho_{p^{b}}}{|  V_{A}^{b}|^{2}}\sigma_{p^{b}}^{2}\left(0.43\mu_{R^{ab}}^{2}-0.5\mu_{R^{ab}}\mu_{X^{ab}}-0.43\mu_{X^{ab}}^2\right) + \\
\frac{\rho_{p^{c}}}{| V_{A}^{c}|^{2}}\sigma_{p^{c}}^{2}\left(-0.43\mu_{R^{ac}}^{2}-0.5\mu_{R^{ac}}\mu_{X^{ac}}+0.43\mu_{X^{ac}}^2\right) + \\
\frac{\rho_{q^{b}}}{| V_{A}^{b}|^{2}}\sigma_{q^{b}}^{2}\left(-0.43\mu_{R^{ab}}^{2}+0.5\mu_{R^{ab}}\mu_{X^{ab}}+ 0.43\mu_{X^{ab}}^2\right) + \\
\frac{\rho_{q^{c}}}{| V_{A}^{c}|^{2}}\sigma_{q^{c}}^{2}\left(0.43\mu_{R^{ac}}^{2}+
0.5\mu_{R^{ac}}\mu_{X^{ac}}-0.43\mu_{X^{ac}}^2\right) + \\ 
\frac{\rho_{p^{a}q^{a}}}{| V_{A}|^{2}}\sigma_{p^{a}}\sigma_{q^{a}}\left(-\mu_{R^{aa}}^{2}+\mu_{X^{aa}}^{2} \right)+\\
\frac{\rho_{p^{b}q^{b}}}{| V_{A}^{b}|^{2}}\sigma_{p^{b}}\sigma_{q^{b}}\left(0.5\mu_{R^{ab}}^{2}+\sqrt{3}\mu_{R^{ab}}\mu_{X^{ab}}-0.5\mu_{X^{ab}}^2\right) + \\
\frac{\rho_{p^{c}q^{c}}}{| V_{c}^{a}|^{2}}\sigma_{p^{c}}\sigma_{q^{c}}\left(0.5\mu_{R^{ac}}^{2}-\sqrt{3}\mu_{R^{ac}}\mu_{X^{ac}}-0.5\mu_{X^{ac}}^2\right)
\end{split}
\label{eq:18b}
\end{equation}
where, $\rho_{p^{h}}$ and $\rho_{q^{h}}$ denote the correlation coefficients of active power and reactive power change among the same phase of different actor nodes with $h$ representing the corresponding phase term ($h=\{a,b,c\}$), respectively. $\rho_{p^{h}q^{h}}$ denotes the correlation coefficient between the active and reactive power within the same phase. Similarly, $\sigma^{2}_{p^{h}}$ and $\sigma^{2}_{q^{h}}$ depict the variance of active power and reactive power change, respectively. For random impedance part, $\mu_{R^{k}}$ and $\mu_{X^{k}}$ denote the mean of shared path resistance and reactance between all the nodes and a certain observation node, respectively, with $k$ representing the corresponding self/mutual impedance terms ($k={aa,ab,ac,ba,bb,bc,ca,cb,cc}$). It is important to note that all the defined parameters with respect to power change are user defined and usually set based on historical data, whereas, the parameters corresponding to shared path impedance are computed based on the network specifications.
\vspace{0.2cm}\\
\textit{\textbf{Step 3:} Compute covariance between $\Delta V_{OA1}^{a,(r,i)}$ and $\Delta V_{OA2}^{a,(r,i)}$}:\\
The covariance between the real component of complex
voltage change caused by two different PVs
located at actor nodes $A1$ and $A2$ can be calculated as:
\begin{equation}
\footnotesize
\begin{split}
  Cov(\Delta V_{OA1}^{a,r}, \Delta V_{OA2}^{a,r}) &=  E(\Delta V_{OA1}^{a,r} \Delta V_{OA2}^{a,r})-E(\Delta V_{OA1}^{a,r})E(\Delta V_{OA2}^{a,r}) \\
  &= E[\bm{Z_{A1}^{(a,r)^{T}}\Delta S_{A1} Z_{A2}^{(a,r)^{T}}\Delta S_{A2}}]
\end{split}
\label{eq:18}
\end{equation}
Using equation (\ref{eq:5}), $\bm{Z^{(a,r)^{T}}\Delta S}$ can be expanded for both the actor nodes in the following way,
\begin{equation*}
\tiny
\begin{split}
{\Huge E}
\Bigg[ \Bigg.
& \Bigg\{ \Bigg.
\frac{\Delta P_{1}^{a}}{|V_{1}^{a}|}(-R_{O1}^{aa}) + \frac{\Delta P_{1}^{b}}{| V_{1}^{b}|}\left(\frac{R_{O1}^{ab}}{2}-\frac{\sqrt{3}X_{O1}^{ab}}{2}\right) +
\frac{\Delta P_{1}^{c}}{| V_{1}^{c}|}\left(\frac{R_{O1}^{ac}}{2}+\frac{\sqrt{3}X_{O1}^{ac}}{2}\right)  \\ &
-\frac{\Delta Q_{1}^{a}}{| V_{1}^{a}|}X_{O1}^{aa} + 
\frac{\Delta Q_{1}^{b}}{| V_{1}^{b}|}\left(\frac{\sqrt{3}R_{O1}^{ab}}{2}+\frac{X_{O1}^{ab}}{2}\right) + 
\frac{\Delta Q_{1}^{c}}{| V_{1}^{c}|}\left(-\frac{\sqrt{3}R_{O1}^{ac}}{2}+\frac{X_{O1}^{ac}}{2}\right) 
\Bigg. \Bigg\} \times \\ 
& \Bigg\{ \Bigg.
\frac{\Delta P_{2}^{a}}{|V_{2}^{a}|}(-R_{O2}^{aa}) + \frac{\Delta P_{2}^{b}}{|V_{2}^{b}|}\left(\frac{R_{O2}^{ab}}{2}-\frac{\sqrt{3}X_{O2}^{ab}}{2}\right) +
\frac{\Delta P_{2}^{c}}{| V_{2}^{c}|}\left(\frac{R_{O2}^{ac}}{2}+\frac{\sqrt{3}X_{O2}^{ac}}{2}\right)  \\ &
-\frac{\Delta Q_{2}^{a}}{| V_{2}^{a}|}X_{O2}^{aa} + 
\frac{\Delta Q_{2}^{b}}{| V_{2}^{b}|}\left(\frac{\sqrt{3}R_{O2}^{ab}}{2}+\frac{X_{O2}^{ab}}{2}\right) + 
\frac{\Delta Q_{2}^{c}}{| V_{2}^{c}|}\left(-\frac{\sqrt{3}R_{O2}^{ac}}{2}+\frac{X_{O2}^{ac}}{2}\right)
\Bigg. \Bigg\} 
\Bigg. \Bigg]
\end{split}
\label{eq:19}
\end{equation*}
For simplicity, actor nodes $A1$ and $A2$ are denoted by subscript $1$ and $2$, respectively. Similar to (\ref{eq:18b}), the terms inside the expectation operator can be cross multiplied to express covariance as,
\begin{equation}
\footnotesize
\begin{split}
\frac{\rho_{p^{a}}}{|V_{A}^{a}|^{2}}\sigma_{p^{a}}^{2} \mu_{R^{aa}}^{2} + \frac{\rho_{q^{a}}}{|V_{A}^{a}|^{2}} \sigma_{q^{a}}^{2}\mu_{X^{aa}}^{2} + \\ 
\frac{\rho_{p^{b}}}{| V_{A}^{b}|^{2}}\sigma_{p^{b}}^{2}\left(0.25\mu_{R^{ab}}^{2}-0.86\mu_{R^{ab}}\mu_{X^{ab}}+0.75\mu_{X^{ab}}^2\right) + \\
\frac{\rho_{p^{c}}}{| V_{A}^{c}|^{2}}\sigma_{p^{c}}^{2}\left(0.25\mu_{R^{ac}}^{2}+0.86\mu_{R^{ac}}\mu_{X^{ac}}+0.75\mu_{X^{ac}}^2\right) +\\
\frac{\rho_{q^{b}}}{| V_{A}^{b}|^{2}}\sigma_{q^{b}}^{2}\left(0.75\mu_{R^{ab}}^{2}+0.86\mu_{R^{ab}}\mu_{X^{ab}}+ 0.25\mu_{X^{ab}}^2\right) + \\
\frac{\rho_{q^{c}}}{| V_{A}^{c}|^{2}}\sigma_{q^{c}}^{2}\left(0.75\mu_{R^{ac}}^{2}-
0.86\mu_{R^{ac}}\mu_{X^{ac}}+0.25\mu_{X^{ac}}^2\right) - \\ 
\frac{\rho_{p^{a}q^{a}}}{| V_{A}^{a}|}\sigma_{p^{a}}\sigma_{q^{a}}\left(2\mu_{R^{aa}}\mu_{X^{aa}}\right)-\\
\frac{\rho_{p^{b}q^{b}}}{| V_{A}^{b}|^{2}}\sigma_{p^{b}}\sigma_{q^{b}}\left(0.86\mu_{R^{ab}}^{2}-\mu_{R^{ab}}\mu_{X^{ab}}-0.86\mu_{X^{ab}}^2\right) + \\
\frac{\rho_{p^{c}q^{c}}}{| V_{c}^{c}|^{2}}\sigma_{p^{c}}\sigma_{q^{c}}\left(-0.86\mu_{R^{ac}}^{2}-\mu_{R^{ac}}\mu_{X^{ac}}+0.86\mu_{X^{ac}}^2\right) 
\end{split}
\label{eq:20}
\end{equation}
The correlation coefficients and variances are same as defined in equation (\ref{eq:18b}).
Now, following the same steps from (\ref{eq:18}) to (\ref{eq:20}), yields the covariance between two actor nodes for the imaginary part of voltage change.
\vspace{0.2cm}\\
\textit{\textbf{Step 4:} Compute mean and variance of $\Delta V_{OA}^{a,r}$ and $\Delta V_{OA}^{a,i}$ due to randomly distributed multiple actor nodes}:
\vspace{0.1cm}\\
The mean value of real and imaginary parts of voltage change due to randomly distributed multiple actor nodes can be written as,
\begin{equation}
\begin{split}
E[\Delta V_{O}^{a,r}]= \mathlarger{\mu_{r}} = E\sum_{A=1}^N \Delta V_{OA}^{a,r}=N\mathlarger{ \bm{\mu_{Z_{o}^{r}} \mu_{\Delta  S}}} \\
E[\Delta V_{O}^{a,i}]= \mathlarger{\mu_{i}} = E\sum_{A=1}^N \Delta V_{OA}^{a,i}= N\mathlarger{ \bm{\mu_{Z_{o}^{i}} \mu_{\Delta S}}} 
\end{split}
\label{eq:23}
\end{equation}
Further, the variance of real and imaginary parts of the net voltage change can be expressed as,
\begin{equation}
\begin{split}
\text{Var}[\Delta V_{O}^{a,r}] = \sigma_{r}^{2} = Var\sum_{A=1}^N (Z_{A}^{(a,r)^T}\Delta S) \\
= N Var(\bm{Z^{(a,r)^T}\Delta S}) + 2\sum_{I<J}Cov(\Delta V_{OI}^{a,r},\Delta V_{OJ}^{a,r}) \\
\text{Var}[\Delta V_{O}^{a,i}] = \sigma_{i}^{2} = 
Var\sum_{A=1}^N (Z_{A}^{(a,i)^T}\Delta S) \\
=N Var(\bm{Z^{(a,i)^T}\Delta S}) + 2\sum_{I<J}Cov(\Delta V_{OI}^{a,i},\Delta V_{OJ}^{a,i})
\end{split}
\label{eq:24}
\end{equation}
Now, by invoking Lindeberg-Feller central limit theorem, it can be shown that the real and imaginary parts of voltage change follow non zero mean Gaussian distribution with mean and variance as stated in equations (\ref{eq:23}) and (\ref{eq:24}), respectively. As the square of non zero mean Gaussian variable follows non-central chi-square distribution \cite{mathai1992quadratic}, the distribution of the squared magnitude of $\Delta V_{O}^{a}$ is the sum of dependent non-central chi-square variables.
\begin{equation}
|\Delta V_{O}^{a}|^{2} \sim \sigma_{r}^{2}\chi_{1}^{2}(\mu_{r}^{2}) + \sigma_{i}^{2}\chi_{1}^{2}(\mu_{i}^{2})
\end{equation}
where $\sigma^{2}$ and $\mu^{2}$ are the weight and non centrality parameters of non central chi square distribution with one degree of freedom corresponding to both real and imaginary parts of the voltage change. The sum of weighted non-central chi-square distributions can then be approximated with a scaled non-central chi-square with weight $\lambda$, non-centrality parameter
$w$, and $v$ degrees of freedom as shown below \cite{mathai1992quadratic}:
\begin{equation}
\left|\Delta V_{O}^{a}\right|^{2} \sim \lambda \chi_{v}^{2}(w)     
\end{equation}
where,
\begin{equation}
\begin{split}
\lambda &=\frac{\sigma_{r}^{4}\left(1+2 \mu_{r}^{2}\right)+\sigma_{i}^{4}\left(1+2 \mu_{i}^{2}\right)}{\sigma_{r}^{2}\left(1+2 \mu_{r}^{2}\right)+\sigma_{i}^{2}\left(1+2 \mu_{i}^{2}\right)}\\
w&=\frac{\left(\sigma_{r}^{2} \mu_{r}^{2}+\sigma_{i}^{2} \mu_{i}^{2}\right)\left(\sigma_{r}^{2}+\sigma_{i}^{2}+2 \sigma_{r}^{2} \mu_{r}^{2}+2 \sigma_{i}^{2} \mu_{i}^{2}\right)}{\sigma_{r}^{4}+\sigma_{i}^{4}+2 \sigma_{r}^{4} \mu_{r}^{2}+2 \sigma_{r}^{4} \mu_{i}^{2}} \\
v&=\frac{\left(\sigma_{r}^{2}+\sigma_{i}^{2}\right)\left(\sigma_{r}^{2}+\sigma_{i}^{2}+2 \sigma_{r}^{2} \mu_{r}^{2}+2 \sigma_{i}^{2} \mu_{i}^{2}\right)}{\sigma_{r}^{2}+\sigma_{i}^{2}+2\left(\sigma_{r}^{4} \mu_{r}^{2}\right)+2\left(\sigma_{i}^{4} \mu_{i}^{2}\right)}
\end{split}
\label{eq:26}
\end{equation}
Since the square root of a non-central chi-square random variables follows a Rician distribution \cite{mathai1992quadratic}, the magnitude of voltage change will follow a Rician distribution:
\begin{equation}
|\Delta V_{O}^{a}| \sim Rician(k, \sigma) 
\label{eq:27}
\end{equation}
where $k=\sqrt{w}$ and $\sigma = \sqrt{\lambda}$. The magnitudes of voltage changes for other phases follow a similar expression with the respective phase values. 

If the power variation is assumed to follow a zero-mean Gaussian distribution, which is a typical assumption used in many prior works \cite{hassanzadeh2010practical, vasilj2015pv, jhala2019data}, $\bm{\mu_{\Delta S}}$ vanishes from the mean (eqn. \ref{eq:9}) and variance (eqn. \ref{eq:15}-\ref{eq:16}) equations of voltage change. This eventually leads to zero value for $\mu_{r}$ and $\mu_{i}$. Again, by invoking Lindeberg-Feller central limit theorem, one can show that the real and imaginary parts of the voltage change follow zero-mean normal distributions as, 
\begin{equation}
\begin{split}
\Delta V_{O}^{a, r} \overset{D}{\sim} \mathcal{N} (0, \sigma_{r}^{2}) \\
\Delta V_{O}^{a, i} \overset{D}{\sim} \mathcal{N} (0, \sigma_{i}^{2})
\end{split}
\label{eq:28a}
\end{equation}
The square of the magnitude of voltage change follows a gamma distribution \cite{lancaster2005chi, chuang2012approximated}, and subsequently, the magnitude of voltage change follows a Nakagami distribution \cite{nakagami1960m},
\begin{equation}
|\Delta V_{O}^{a}| \sim Nakagami(m, \omega),
\label{eq:28}
\end{equation}
where parameter $\theta = 2(\sigma_{r}^4 + \sigma_{i}^4+2c^2) /(\sigma_{r}^2 + \sigma_{i}^2) $, shape parameter $m=(\sigma_{r}^2 + \sigma_{i}^2) /\theta$, scale parameter $\omega=\sqrt{m\theta}$, and $c$ being the covariance between the real and imaginary parts of voltage change.
In the next sections, the proposed ST-PVSA method is first validated using simulations, and then it is employed to estimate PV HC in a efficient manner. 

\section{Validation of ST-PVSA}

The proposed probability distribution of the voltage change is validated on the modified IEEE 37-node test system. The nominal voltage of the test system is 4.8 kV. The actual distribution of the magnitude of the voltage change is obtained using Newton-Raphson based sensitivity analysis method, and the theoretical distribution is obtained using the proposed method of ST-PVSA. A scenario is considered for simulation where $9$ PV units are located at random locations in the distribution system. The power at the actor nodes, i.e., the nodes injected with PVs, varies randomly due to fluctuations in PV generation. In practice, PV generation can be modeled as a time series with a trend component along with some random noise typically modeled as a Gaussian random variable \cite{hassanzadeh2010practical, vasilj2015pv, jhala2019data}. Consistent with these prior efforts, in this work, change in PV generation at a particular time instant is modeled as a zero mean Gaussian random variable. Further,
the covariance matrix $\textstyle\sum_{\Delta S}$ captures the spatial correlation of PV generation, and can be estimated in practice based on a number of factors. The spatial correlation exists because of
geographical proximity as PVs in the same region typically exhibit same generation profile. The diagonal elements of the covariance matrix contain variances that depend on the size of PV units and the off-diagonal elements capture the effect of geographical proximity of these PV units. Specifically, in our simulation, the variance of
change in real power ($\Delta P$) is set to $5$ kW and the variance
of change in reactive power ($\Delta Q$) is set to $0.5$ kVar.
Further, the values of the correlation coefficients $\rho_{p^{h}}$, $\rho_{q^{h}}$ and $\rho_{p^{h}q^{h}}$ are set to $0.2$, $0.2$, and $-0.5$, respectively for all the phases. Variance can be set to zero for nodes with no PVs. Now, for random impedance part, the mean and variance of resistance and reactance between random actor node and observation node $9$ is calculated from data of the IEEE 37-node test system. The Value of correlation coefficient between resistance
and reactance is $0.99$. 
\begin{figure}[t]
	\includegraphics[width = 8.0cm,height=3.7cm]{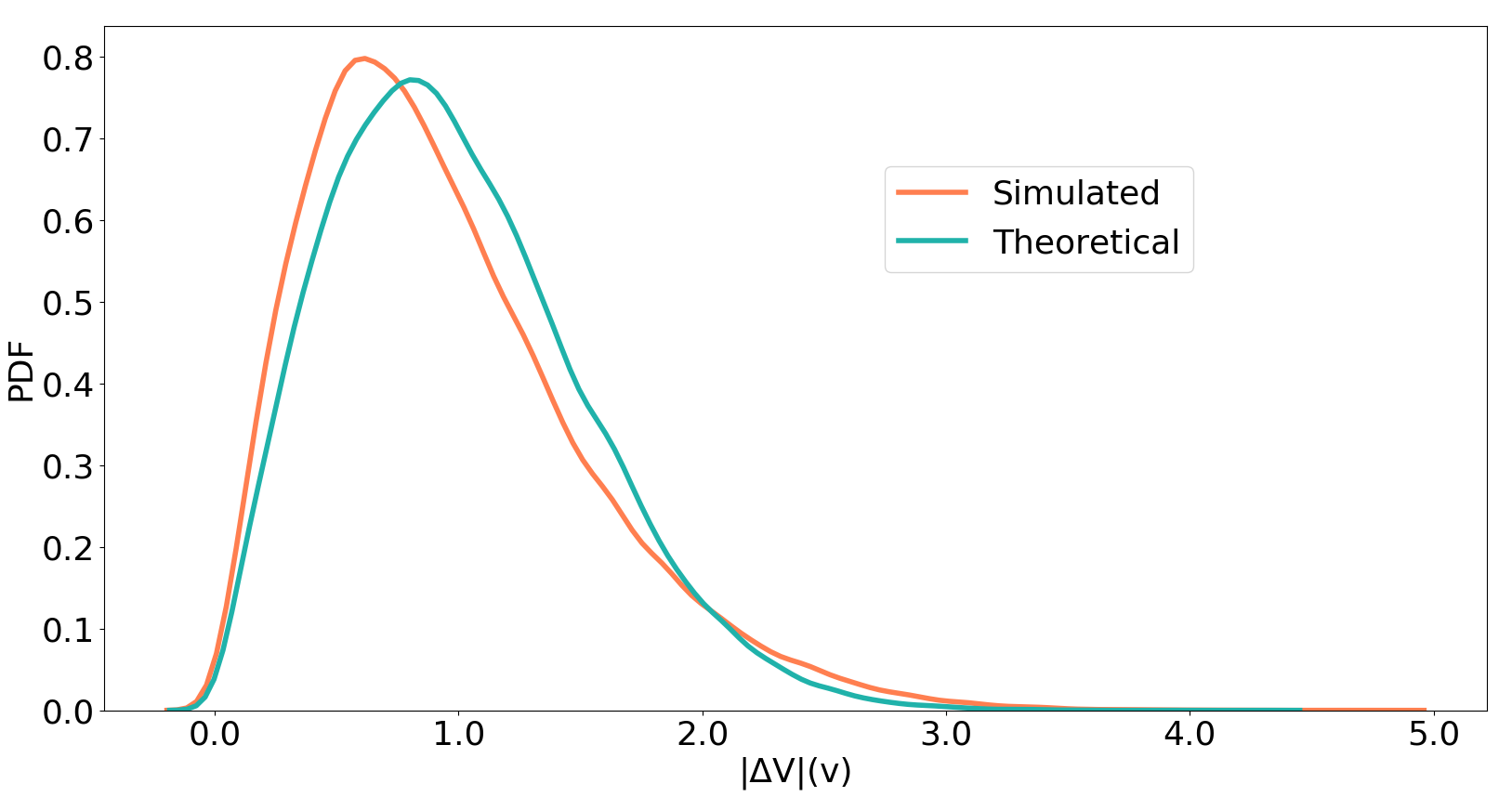}
	\caption{Distribution of voltage change at node 9 }
	\label{fig:2}
\end{figure}
Fig. \ref{fig:2} compares the actual distribution of the magnitude of voltage change with the proposed theoretical case. The actual distribution of $\Delta V_{9}$ is obtained by randomly varying powers of all actor nodes at phase-a and subsequently, voltage change at node $9$ is computed by using Newton-Raphson based method. Further, Monte-Carlo simulations (MCS) are incorporated to capture the uncertainties associated with the power changes. Here, voltage changes are computed for one million Monte-Carlo simulations. The scaled histogram of $|\Delta V_{9}|$ is depicted through the orange curve in the figure \ref{fig:2}. The theoretical distribution computed with equations (\ref{eq:23}) and (\ref{eq:24}) is shown by blue curve in Fig. \ref{fig:2}. It can be observed that the probability distribution computed using the proposed method is very close to the actual simulated distribution with $0.18$ as Jensen-Shannon distance. This experiment demonstrates the effectiveness of the proposed ST-PVSA approach.
\section{ST-PVSA for PV hosting capacity}
This section presents the methodology for computing HC with the proposed ST-PVSA approach. As ST-PVSA provides the probability distribution of voltage change at a node due to random power changes at random locations of the network, it suffices to identify voltage violations for different PV penetration levels. The procedure to determine HC begins with fixing the number of penetration levels. Then, the number of PV units ($N_{k}$) integrated into the system is computed for each penetration level. Using equation (\ref{eq:29}), $N_{k}$ is determined statistically based on the distribution of PV size from California dataset \cite{CaliforniaSolar}. The penetration level is divided into various bands based on the percentage of total demand. For instance, $k$ varies from $1$ to $5$ for $5$ bands, i.e., $(0-20 \%),(21-40 \%)$ ... $(81-100 \%)$. A unique $N_{k}$ is defined for each band such that the same number of PV units is used for all penetration levels in that particular band. This is logical in the sense that it is not necessary to increase the number of PV units for simulating increasing penetration level rather it can be achieved by increasing the power injection in the existing PVs. However, the power injections cannot be increased beyond a certain limit due to the restriction of PV size. Therefore, $N_{k}$ increases as we move to the higher penetration band. $N_{k}$ for a particular penetration band $k$ is computed as following:
\begin{figure}[t]
\centering
	\includegraphics[width = 8.5cm, height=6.0cm]{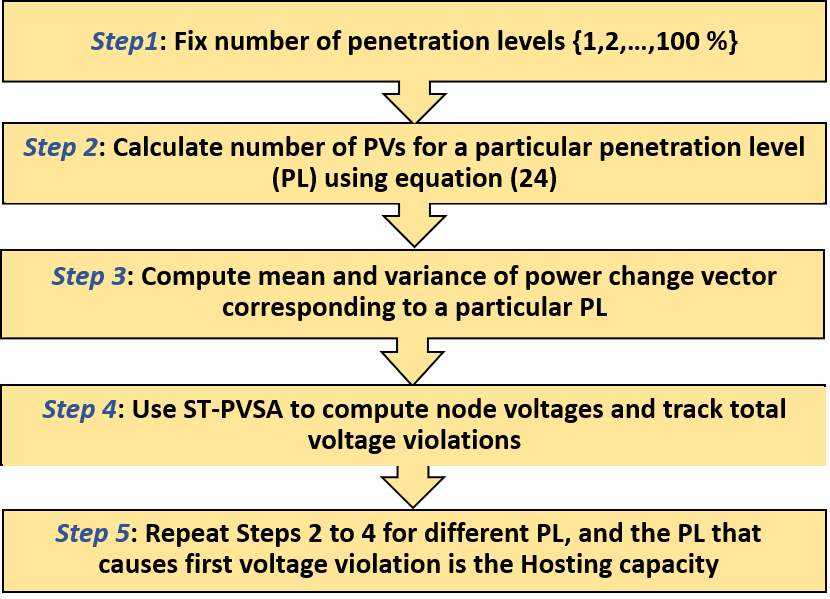}
	\caption{Flowchart of ST-PVSA approach for HC }
	\label{fig:4}
\end{figure}
\begin{equation}
    N_{k}= \frac{\text{Mean penetration level (power injection) for band $k$}}{\text{Max PV size}}
    \label{eq:29}
\end{equation}
where \quotes{Max PV size} comes from the PV size distribution and the mean penetration level is the average power injection for band $k$. Then, in the third step, $N_{k}$ is used to obtain the mean ($\mu_{\Delta S}^{l}$) of power change vector $\Delta S$ (equation \ref{eq:5}) for each penetration level $l$, such that $\mu_{\Delta S}^{l}N^{l}\approx P_{l}$, where $P_{l}$ and $N^{l}$ are the net power injection and PV units for penetration level $l$. The complex voltage change due to power injection is added to the base voltage to get the future voltage. Following the same arguments as mentioned in Theorem $2$ of \cite{abujubbeh2020probabilistic}, the distribution of future voltage is shown to follow Rician with the parameters as defined in equation (\ref{eq:26}). The mean of real ($\mu_{r}$) and imaginary ($\mu_{i}$) parts of voltage change (\ref{eq:23}) are added to the corresponding part of node base voltages to get the mean value of the future voltage, which is then plugged in equation (\ref{eq:27}) to find the distribution of future voltage at all nodes of the network. Nodes with a probability of voltage violation greater than $0.5$ are identified as vulnerable nodes and recorded as a violation. The complete process is repeated for increasing penetration levels until the algorithm encounters the first violation. The corresponding penetration level is the HC of the system. Fig. \ref {fig:4} summarizes the ST-PVSA steps to compute HC. 
To evaluate the performance of ST-PVSA in determining the HC, load flow based HC is used as a benchmark. Similar to the ST-PVSA approach, the PV penetration level is fixed from $1\%$ to $100\%$ level at $1\%$ increment. For each penetration level, Monte Carlo simulations are repeated $10k$ times thereby creating one million different PV deployment scenarios. For illustration purposes, the loads on the test network are chosen as reported in the IEEE PES distribution system analysis subcommittee report. However, the proposed method is generic enough to accommodate other loading scenarios such as daytime (10 am-2 pm) maximum load and daytime minimum.
Finally, for each penetration level, $N_{k}$ locations are selected randomly to allocate PV units and load flow is executed to track the voltage violations. 
Further, the HC analysis is validated on the IEEE 37-node as well as the IEEE 123-node network to demonstrate the scalability of the proposed approach. Table \ref{Table:Hc_value} presents the HC values computed with the proposed ST-PVSA based approach and existing load flow-based method for a various number of scenarios. It is worth noting that for each PV penetration level, ST-PVSA needs to be run once (independent of scenarios), whereas multiple simulations are required for convergence in the load flow-based approach. The estimation accuracy is fairly high with a significantly low computational burden as reflected from the execution times in Table \ref{Table:Hc_Time}. ST-PVSA is an order two faster than the load flow-based approach and the gap will further increase as the network size grows.
\begin{table}[t]
    \centering
    \captionsetup{width=.85\textwidth}
    \caption{Hosting capacity with Load flow and ST-PVSA}
    \label{Table:Hc_value}
	\begin{tabular}{|c|c|c|c|c|}
	\hline
    Test Network & Scenarios\textsuperscript{*} & 1k & 10k & 30k  \\ 
    \cline{1-5}
    \multirow{ 2}{*}{IEEE 37} & Load Flow & 34 & 32 & 31 \\
    \cline{2-5}                      
    & ST-PVSA & \multicolumn{3}{c|}{33} \\
    \hline
    \multirow{ 2}{*}{IEEE 123} & Load Flow & 44 & 43 & 42 \\
    \cline{2-5}  
    & ST-PVSA & \multicolumn{3}{c|}{39} \\
    \hline
  \end{tabular}
\end{table}
\begin{table}[t]
    \centering
    \captionsetup{width=.85\textwidth}
    \caption{Execution time (min) with Load flow and ST-PVSA}
    \label{Table:Hc_Time}
	\begin{tabular}{|c|c|c|c|c|}
	\hline
    Test Network & Scenarios\textsuperscript{*} & 1k & 10k & 30k  \\ 
    \cline{1-5}
    \multirow{ 2}{*}{IEEE 37} & LF & 1.93 & 18.8 & 55.7 \\
    \cline{2-5}                      
    & ST-PVSA & \multicolumn{3}{c|}{1.13} \\
    \hline
    \multirow{ 2}{*}{IEEE 123} & LF & 43.93 & 400.2 & 1195.6 \\
    \cline{2-5}  
    & ST-PVSA & \multicolumn{3}{c|}{3.91} \\
    \hline
  \end{tabular}
\end{table}

\section{Conclusion}

This work presents an analytical approach to compute the probability distribution of voltage change at a particular node as a function of random change in power at random locations of the network due to distributed PV units. This method is then employed to determine the HC of the system without the need to investigate multiple scenarios. The proposed ST-PVSA approach is computationally efficient and yet accurate. In the IEEE 123-node test network, ST-PVSA based HC is an order two faster compared to conventional load flow based approach and this gap will further increase as the network size grows. As part of future work, we plan to extend ST-PVSA for dynamic HC involving continuous time series data of load and PV with the system allowing violations of small duration in accordance with the ANSI standards.
\vspace{-0.2cm}
\bibliographystyle{IEEEtran}
\bibliography{IEEEabrv, PVSA_ST_arxiv}
\end{document}